\documentclass[11pt]{article}
\usepackage{algorithm}
\usepackage{algpseudocode}
\usepackage{amsfonts}
\usepackage{amsmath}
\usepackage{amsthm}

\usepackage[backref=page]{hyperref}

\usepackage[capitalize,noabbrev,nameinlink]{cleveref}

\usepackage{float}
\usepackage{fullpage}
\usepackage{amssymb}
\usepackage[dvipsnames]{xcolor}
\usepackage{graphicx}
\usepackage{float}
\usepackage{dsfont}
\usepackage{bbm}
\usepackage{bm}
\usepackage{booktabs}
\usepackage{mathtools}
\usepackage{enumitem}
\usepackage{thmtools}
\usepackage{setspace}
\usepackage[explicit]{titlesec}

\usepackage[labelfont=bf]{caption}

\definecolor{corlinks}{RGB}{200,0,0}
\hypersetup{
    colorlinks = true,
    urlcolor = myred,
    linkcolor = myred,
    menucolor = cormenu,
    citecolor = myred,
    pdfborder = 0 0 0
}

\theoremstyle{plain}
\newtheorem{theorem}{Theorem}
\newtheorem{lemma}{Lemma}

\theoremstyle{definition}
\newtheorem{definition}{Definition}

\theoremstyle{remark}
\newtheorem{claim}{Claim}

\definecolor{myred}{rgb}{0.64, 0.0, 0.0}

\newcommand{\eps}{\varepsilon}
\DeclareMathOperator{\supp}{supp}

\DeclareMathOperator{\pluck}{pluck}
\newcommand{\PM}{\textsc{Match}}
\newcommand{\Odd}{\textsc{Odd}}

\definecolor{darkred}{rgb}{0.55, 0.0, 0.0}
\colorlet{darkdarkred}{darkred!70!black}
\definecolor{bisque}{rgb}{1.0, 0.84, 0.69}
\colorlet{darkbisque}{bisque!70!black}
\definecolor{burntumber}{rgb}{0.54, 0.2, 0.14}
\definecolor{sangria}{rgb}{0.57, 0.0, 0.04}
\usepackage{tikz}

\makeatletter
\newcommand{\numletter}[1]{\@alph{#1}}
\makeatother

\usetikzlibrary{
    babel,
    arrows,
	calc,
	patterns,
	intersections,
    backgrounds,
    trees,
    mindmap,
    fadings,
	decorations.pathreplacing,
	decorations.pathmorphing,
	decorations.text,
	decorations.markings,
    scopes,
	shapes,
    shapes,
	positioning
}

\tikzfading[name = middle,
    top color = transparent!100,
    bottom color = transparent!100,
    middle color = transparent!30
]

\tikzset{
    >=latex,
    graph-vert/.style 2 args = {
        draw,
        color = #1!70!black,
        circle,
        thick,
        inner sep = 0pt,
        minimum size = #2,
        fill = #1!50,
        fill opacity = 0.5
    },
    graph-vert/.default = {blue}{0.2cm},
    triangle/.style 2 args = {
        #1!70!black,
        fill = #1!30,
        fill opacity = #2
    },
    universe-rect/.pic = {
        \fill[black!10, rounded corners = 3pt] (-0.1, 0.1) rectangle (#1 + 0.1, -#1 - 0.1);
        \fill[white] (0, 0) rectangle (#1, -#1);
        \draw[step = 0.2, black!10] (0, 0) grid (#1, -#1);
        \draw[thick] (-0, 0) rectangle (#1, -#1);
    },
    pics/shaded-area/.style n args = {3}{
        code = {
            \draw[#1, opacity = 0.2, fill, fill opacity = 0.1, thick]
                (-0.2, 0) arc (-180:-90:0.2) -- ++(#2, 0) arc (-90:0:0.2) -- ++(0, #3) arc (0:90:0.2)
                -- ++(-#2, 0) arc (90:180:0.2) -- cycle;
        }
    },
    vert/.style = {
        draw,
        thick,
        rectangle,
        rounded corners = 2pt,
    },
    semisim/.style = {
        ->,
        blue,
        dashed,
        decorate,
        decoration = {
            snake,
            amplitude = 0.5,
            segment length = 2
        }
    }
}

\newcommand{\complexityfont}[1]{\mathsf{#1}}

\def\RNC          {\complexityfont{RNC}}

\def\P            {\complexityfont{P}}

\def\AC           {\complexityfont{AC}}
\def\pL           {\oplus\L}

\def\NC           {\complexityfont{NC}}
\def\SPL          {\complexityfont{SPL}}
\renewcommand{\L} {\complexityfont{L}}

\def\cald{{\mathcal D}}

\def\calg{{\mathcal G}}

\def\calm{{\mathcal M}}

\def\calp{{\mathcal P}}

\def\cals{{\mathcal S}}

\def\calv{{\mathcal V}}

\title{
\huge Monotone Circuit Complexity of Matching
}

\author{
\begin{tabular}{ccc}
Bruno Cavalar &
Mika G\"o\"os&
Artur Riazanov \\[-1mm]
\small\slshape University of Oxford &
\small\slshape EPFL &
\small\slshape EPFL 
\end{tabular}\\[1em]  
\begin{tabular}{cc}
Anastasia Sofronova &
Dmitry Sokolov \\[-1mm]
\small\slshape EPFL &
\small\slshape EPFL \& Universit{\'{e}} de Montr{\'{e}}al
\end{tabular}  
}

\date{\today}

\begin{document}

\maketitle

\begin{abstract}
\noindent
We show that the perfect matching function on $n$-vertex graphs requires monotone circuits of size $\smash{2^{n^{\Omega(1)}}}$. This improves on the 
$n^{\Omega(\log n)}$
lower bound  of Razborov (1985). Our proof uses the standard approximation method together with a new sunflower lemma for matchings.
\end{abstract}

\section{Introduction}

A sobering lesson learned already in the 1980s~\cite{Razborov1985,Tardos1988} is that general boolean circuits (using gates $\land$, $\lor$, $\neg$) can be much more powerful than monotone circuits (using gates~$\land$,~$\lor$). The earliest demonstration is due to Razborov~\cite{Razborov1985}. He considered the \emph{bipartite perfect matching} function $\PM\colon\{0,1\}^{n^2}\to\{0,1\}$ that takes as input a bipartite graph, represented by its adjacency matrix $x\in\{0,1\}^{n\times n}$, and outputs $\PM(x)= 1$ iff the graph contains a perfect matching. While bipartite matching famously admits polynomial-size circuits, Razborov showed that it requires monotone circuits of size $n^{\Omega(\log n)}$. Since then, a long-standing challenge has been to determine whether Razborov's quasi-polynomial bound is tight (e.g., see textbooks~\cite{Wegener1987,Jukna2012,Wigderson2019}). Our main result is to improve the lower bound to an exponential one.
\begin{theorem} \label{thm:main}
$\PM$ requires monotone circuits of size at least $2^{n^{1/3-o(1)}}$.
\end{theorem}
That is, for bipartite matching, the gap between the general and monotone circuit complexities is exponential. In fact, such an exponential gap was already known for a different monotone function in class $\P$ due to Tardos~\cite{Tardos1988}. Her function is relatively complex, however, as it is computed by solving a semidefinite program. Meanwhile, bipartite matching admits an efficient parallel algorithm (class $\RNC$)~\cite{Lovasz1979,Mulmuley1987}, which is not known for Tardos's function.

Another serious contender for exhibiting the strongest general-vs-monotone
separation is \emph{$\mathbb{Z}_2$-satisfiability}~\cite{Goos2019}. This is
a monotone function encoding the problem ``Given a system of linear
equations over $\mathbb{Z}_2$, is it satisfiable?'' It is complete for the class $\pL$ of problems computed by uniform polynomial-size parity branching programs~\cite{Damm1990}. Yet, $\mathbb{Z}_2$-satisfiability was shown to require exponential-size monotone circuits by~\cite{Garg2020,Goos2019}. Bipartite matching is not known to be comparable to $\mathbb{Z}_2$-satisfiability under deterministic reductions. However, \emph{non-uniformly}, bipartite matching lies in a subclass~$\SPL\subseteq\pL$~\cite{Allender1999} and hence is arguably simpler than $\mathbb{Z}_2$-satisfiability.

In fact, our proof of \cref{thm:main} can be extended further to prove a lower bound for a function even
simpler than bipartite matching, called \emph{odd factor}~\cite{Babai1999}. This function is defined
by~$\Odd(x)= 1$ iff the graph $x\in\{0,1\}^{n\times n}$ contains a spanning subgraph whose degrees are
all odd. Equivalently,~$\Odd(x)= 1$ iff every connected component of $x$ has even size. This function can
be computed in logarithmic space (class $\L$) using Reingold's algorithm~\cite{Reingold2008}. We show
monotone lower bounds for $\Odd$ as well as ``padded'' versions of it that can be computed by one of the
simplest of all circuit models: constant-depth circuits (class $\AC^0$).
\begin{theorem}
    \label{thm:l}
    There is a monotone $\Odd\in\L$ with monotone circuit complexity $2^{n^{\Omega(1)}}$.
\end{theorem}

\begin{theorem}
    \label{thm:ac0}
    For any $k$ there is a monotone $f_k\in\AC^0$ with monotone circuit complexity $n^{\Omega(\log^k n)}$.
\end{theorem}

In particular, \cref{thm:ac0} resolves an open problem of Grigni and Sipser~\cite{Grigni1992}, who asked
if every monotone function in $\AC^0$ can be computed by a polynomial-size monotone circuit.
\Cref{thm:ac0} also rules out a particular approach to obtaining general circuit lower bounds: the papers~\cite{Chen2022,Cavalar2023} observed that if \Cref{thm:ac0} had turned out to be false, then $\NC^2 \not\subseteq \NC^1$.

\begin{table}[t]
\centering
\renewcommand{\arraystretch}{1}
\newcommand{\pad}{\hskip 1.5em}
\setlength\tabcolsep{.2em}
\begin{tabular}{l @{\pad} l @{\pad} l @{\pad} r}
\toprule[.5mm]
\bf Function & \bf Class & \bf Monotone complexity & \bf Reference \\
\midrule
Bipartite matching & $\RNC$ & $\exp(\Omega(\log^2 n))$ & \cite{Razborov1985} \\
Tardos's function & $\P$ & $\exp(n^{\Omega(1)})$ & \cite{Alon1987,Tardos1988} \\
Odd factor & $\L$ & $\exp(\Omega(\log^2 n))$ & \cite{Babai1999} \\
$\mathbb{Z}_2$-satisfiability &  $\pL$ & $\exp(n^{\Omega(1)})$ & \cite{Garg2020,Goos2019} \\
$\mathbb{Z}_2$-satisfiability, padded & $\AC^0[\oplus]$ & $\exp(\Omega(\log^k n))$ & \cite{Cavalar2023} \\
\midrule
Bipartite matching & $\RNC$ & $\exp(n^{\Omega(1)})$ & This work (\cref{thm:main}) \\
Odd factor & $\textsf{L}$ & $\exp(n^{\Omega(1)})$ & This work (\cref{thm:l}) \\
Odd factor, padded & $\AC^0$ & $\exp(\Omega(\log^k n))$ & This work (\cref{thm:ac0})\\
\bottomrule[.5mm]
\end{tabular}
\caption{Timeline of separations between general and monotone complexities.
The parameter $k$ can be taken to be any large constant at the cost of
increasing the depth of the $\AC^0$ circuit.}
\label{tab:separations}
\end{table}

\subsection{Technique: Matching sunflowers}
\label{sec:techniques}

We follow the classic \emph{approximation method} introduced by Razborov~\cite{Razborov1985,Razborov1985a}. By now, this standard method is featured in several textbooks~\cite{Wegener1987,Arora2009,Jukna2012,Watson2025}. To prove a lower bound for an $n$-bit boolean function $f$, the method starts by defining a distribution~$\cald$ over the input domain $\{0,1\}^n$. The goal is to show that (i) if $f$ is computed by a small monotone circuit, then $f$ can be approximately computed (relative to $\cald$) by a small monotone DNF; and (ii)~no small DNF correlates well with $f$.

The technical crux of the proof is to identify situations when a monotone DNF $\bigvee_{S\in\cals} t_S$ (where $t_S\coloneqq \bigwedge_{i\in S}x_i$ and $\cals\subseteq 2^{[n]}$) can be safely replaced 
with the single term $t_K$ where $K\coloneqq\bigcap\cals$ is the \emph{core}, and doing so does not incur much error (relative to $\cald$). This replacement procedure, often called ``plucking'', simplifies the DNF in case $|\cals|\geq 2$. 
Previous works~\cite{Rossman2014,Cavalar2020,Blasiok2025}
have employed the notion of ``robust sunflowers'' to find such DNFs for plucking. A family $\cals\subseteq 2^{[n]}$, $|\cals|\geq 2$, is an
\emph{$\varepsilon$-robust sunflower} if 
the core $K = \bigcap \cals$
satisfies
\begin{equation}
    \label{eq:robustness-cond}
    \Pr_{\bm x \sim \{0,1\}^n}[\,\exists S \in \cals \colon t_{S
    \setminus K}(\bm x) =1\,] \geq 1-\eps.
\end{equation}
This says that, for a uniform random $\bm x$, whenever $t_K$ accepts $\bm x$, it is highly likely that $\bigvee_{S\in\cals}t_S$ accepts it too. That is, the approximation error is small. Recent works~\cite{Alweiss2021,Rao2020,bcw21} have proved optimal bounds on the size of families $\cals$ that are guaranteed to contain a subfamily $\cals'\subseteq\cals$ that is a robust sunflower.
We use the following bound from~\cite[Theorem 3]{bcw21}\footnote{
    The lemma follows directly from the main result of~\cite{bcw21} 
    by a standard argument~\cite[Appendix A]{Cavalar2020}.
}.
\begin{lemma}[\protect{Robust Sunflower Lemma~\cite{bcw21}}]
\label{lem:robust-sunflowers}
    There exists a universal constant~$c>0$
    such that every
    family $\cals$ of $\ell$-sets of size $|\cals|\geq (c\log (\ell/\varepsilon))^\ell$ contains an $\eps$-robust sunflower.
\end{lemma}

For our purposes, we need instead a sunflower lemma tailored to the bipartite matching problem. The first difference is that, instead of a uniform distribution, we will pluck relative to the following \emph{odd cut} distribution $\cald_0$ over $\PM^{-1}(0)$ (which Razborov~\cite{Razborov1985} also used).
\begin{definition}[Odd cut distribution $\cald_0$]\label{def:odd-cut}
    To sample $\bm x\sim\cald_0$, first sample a uniform random
    colouring $\bm c \in \{0,1\}^{2n}$ of the vertices of $K_{n,n}$ with an odd number of $1$s. To build
    the bipartite graph $\bm x$, connect any two vertices (on opposite sides) that have the same colour under $\bm
    c$. The resulting graph is a union of two odd-sized bicliques:
    
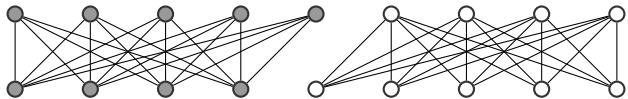
\begin{center}
        \begin{tikzpicture}[scale=1.1]
    \def\n{8}
    \def\size{0.2cm}
    \colorlet{onecolor}{black} 
    \colorlet{zerocolor}{white}

    \foreach \i in {0, 1, ..., \n}{
        \ifnum\i<5
            \edef\col{onecolor}
        \else
            \edef\col{zerocolor}
        \fi
        \node[graph-vert = {\col}{\size}, black!80, fill = \col!80] (u\i) at (\i, 1) {};
    }

    \foreach \i in {0, 1, ..., \n}{
        \ifnum\i<4
            \edef\col{onecolor}
        \else
            \edef\col{zerocolor}
        \fi
        \node[graph-vert = {\col}{\size}, black!80, fill = \col!80] (d\i) at (\i, 0) {};
    }

    \foreach \i in {0, 1, ..., 4}{
        \foreach \j in {0, 1, ..., 3}{
            \draw (u\i) -- (d\j);
        }
    }
    \foreach \i in {5, 6, ..., \n}{
        \foreach \j in {4, 5, ..., \n}{
            \draw (u\i) -- (d\j);
        }
    }
\end{tikzpicture}

    \end{center}
\end{definition}
Suppose $\calm$ is a family of 
$\ell$-matchings (each matching has $\ell$ edges) in $K_{n, n}$. We say that the family $\calm$, $|\calm|\geq2$, is an \emph{$\eps$-matching sunflower} if the core $K = \bigcap \calm$ satisfies
\begin{equation} \label{eq:match-sunflower}
    \Pr[\,\exists M \in \calm,
    \forall e \in M\setminus K :
    \text{$e$ is monochromatic under $\bm c$}\,] \geq 1-\eps.
\end{equation}
This says that, for an input $\bm x\sim\cald_0$, whenever $t_K$ accepts $\bm x$, it is highly likely that $\bigvee_{M\in\calm}t_M$ accepts it too. That is, the approximation error is small. We prove the following in \cref{sec:sunflower}.
\begin{lemma}[Matching Sunflower Lemma]
    \label{lem:matching-sunflowers}
    \label{lem:sunflower-bound}
    There exists a universal constant $c>0$ such that every
    family $\calm$ of $\ell$-matchings of size $|\calm|\geq (c\ell\log^2 (\ell/\varepsilon))^\ell$ contains an $\eps$-matching sunflower.
\end{lemma}
Our proof is surprisingly simple: it is by a \emph{reduction} to the robust sunflower lemma. 
Implicit in the original proof of Razborov~\cite{Razborov1985} is that one
can take $|\calm|\geq 4^{\ell^2}(c \ell \log(1/\eps))^{2\ell}$ in the above
lemma. This is exponentially worse in terms of $\ell$. Our improvement
above is what directly translates into an exponential monotone circuit
lower bound. We discuss in \cref{sec:approximation} how
\cref{lem:matching-sunflowers} plugs into the standard approximation method
to prove \cref{thm:main,thm:l,thm:ac0}.

\subsection{Other related work}

The analogue of \cref{thm:main} for monotone \emph{formulas} has been known
for a long time. Raz and Wigderson~\cite{Raz1992} proved that bipartite
matching requires monotone formulas of size $2^{\Omega(n)}$ and this is
tight. The same lower bound holds for odd factor~\cite{Babai1999} and this
was shown to imply quasi-polynomial lower bounds for the $\AC^0$-computable
padded odd factor in~\cite{Cavalar2023}.

Previous works have also studied general-vs-monotone separations in the setting of constant-depth circuits.
Okol’nishnikova~\cite{Okolnishnikova1982} and Ajtai and
Gurevich~\cite{Ajtai1987} exhibited a monotone function in $\AC^0$ that
requires monotone constant-depth circuits of size $n^{\omega(1)}$. This was
quantitatively improved by Chen, Oliveira, and Servedio~\cite{Chen2017}
showing a lower bound of $2^{n^{\Omega(1)}}$. Our \cref{thm:ac0} thus
improves qualitatively on \cite{Okolnishnikova1982,Ajtai1987,Cavalar2023} but is
incomparable to~\cite{Chen2017}.
It remains open whether there exists a monotone function in $\AC^0$
with exponential monotone circuit complexity.
This would be an ultimate general-vs-monotone separation, generalising all the aforementioned results. The analogous question for \emph{arithmetic} circuits was only recently
settled~\cite{Chat21}.

Besides general-vs-monotone separations, another foremost goal in monotone
complexity is to find explicit functions with maximal monotone circuit
complexities. The clique function has been studied the
most~\cite{Razborov1985,Alon1987,Jukna1999,Rossman2014,Cavalar2020,Lovett2022,Blasiok2025,Rezende2025}.
Currently, the largest explicit lower bound is
$\smash{2^{n^{1/2-o(1)}}}$~\cite{Cavalar2020} with previous records being held
by~\cite{Andreev1987,Harnik2000}. The question of proving ``truly''
exponential lower bounds of the form $2^{\Omega(n)}$ remains open. It has
been solved for monotone \emph{formulas} by Pitassi and
Robere~\cite{Pitassi2017}.

\section{Matching Sunflower Lemma \texorpdfstring{(\cref{lem:sunflower-bound})}{}}
\label{sec:sunflower}

In this section, we prove \cref{lem:sunflower-bound}.
Suppose $\calm$ is a family of $\ell$-matchings with $|\calm| \geq (c \ell \log^2(\ell/\eps))^{\ell}$, where $c$ is a large enough constant (to be determined). Our goal is to show that $\calm$ contains an $\eps$-matching sunflower. We start by simplifying the family $\calm$ by making it ``blocky''.

\paragraph{Reduction to blocky families.}
Consider partitioning the vertices of $K_{n,n}$ into $\ell$ blocks according to a random labelling $\bm b\sim[\ell]^{2n}$.  We say that a
matching $M$ is \emph{consistent} with $\bm b$ if every edge $uv\in M$ is monochromatic under $\bm b$ (that is, $\bm b_u = \bm b_v$) and all edges receive distinct labels. Here is an illustration (with blocks of the same size, for simplicity):
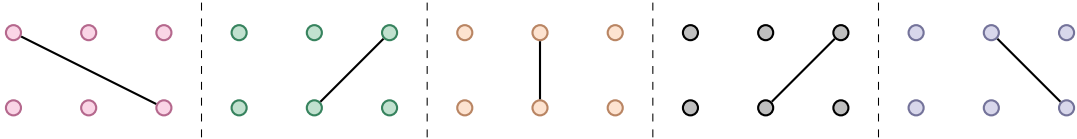
\begin{center}
    \begin{tikzpicture}[scale=1.1]
    \def\colorlist{{
            "VioletRed", "VioletRed", "VioletRed",
            "ForestGreen", "ForestGreen", "ForestGreen",
            "Peach", "Peach", "Peach",
            "black", "black", "black",
            "Periwinkle", "Periwinkle", "Periwinkle",
            "VioletRed", "Sepia", "black", "black", "black"
        }}

    \def\n{14}
    \def\size{0.2cm}

    \foreach \i in {0, 1, ..., \n}{
        \pgfmathparse{\colorlist[\i]}
        \edef\col{\pgfmathresult}
        \node[graph-vert = {\col}{\size}] (u\i) at (\i, 1) {};
        \node[graph-vert = {\col}{\size}] (d\i) at (\i, 0) {};
    }

    \foreach \i in {2.5, 5.5, 8.5, 11.5}{
        \draw[dashed] (\i, 1.4) -- ++(0, -1.8);        
    }

    \draw[thick] (u0) -- (d2);
    \draw[thick] (d4) -- (u5);
    \draw[thick] (u7) -- (d7);
    \draw[thick] (d10) -- (u11);
    \draw[thick] (u13) -- (d14);
\end{tikzpicture}

\end{center}

For a fixed $\ell$-matching $M\in\calm$, there are $\ell^{2\ell}$ ways of labelling its endpoints, and $\ell!$ of these yield a consistent labelling. Thus $
\Pr[\,\text{$M$ is consistent with $\bm b$}\,] = \ell!/\ell^{2\ell} \geq \ell^{-\ell}2^{-O(\ell)}$. By averaging, there exists a fixed labelling $b$ that is consistent with at least $|\calm| \ell^{-\ell}2^{-O(\ell)}$ matchings in $\calm$. Let us delete all matchings inconsistent with $b$, and continue to denote the resulting set by $\calm$ for simplicity. For large enough $c$, the number of remaining matchings is $|\calm| \geq (c' \log (4\ell / \eps))^{2\ell}$, where~$c'$ is the universal constant from the robust sunflower lemma (\cref{lem:robust-sunflowers}).

\paragraph{Finding a ``vertex'' sunflower.}
Define $\calv \coloneqq \{V(M) : M \in \calm\}$, $V(M)\coloneqq \bigcup M$, as the family of endpoints of
matchings in $\calm$. Since $\calm$ is blocky, $\calv$ and $\calm$ are in 1-to-1 correspondence: for
every~$V \in \calv$ there is a \emph{unique} matching $M\in\calm$ with $V=V(M)$. Thus $\calv$ is a family
of $2\ell$-sets of size~$|\calv|=|\calm|\geq (c' \log (4\ell/\varepsilon))^{2\ell}$. We can now apply
\cref{lem:robust-sunflowers} to find an $\eps/2$-robust sunflower~$\calv'\subseteq\calv$ (``vertex''
sunflower) with core $K \coloneqq \bigcap \calv'$. Note that every block contains $0$, $1$, or $2$
vertices from $K$. Let $K_1$ be the vertices in $K$ that are unique in their block, and $K_2$ be the
vertices that share a block with another vertex, so $K = K_1 \sqcup K_2$.

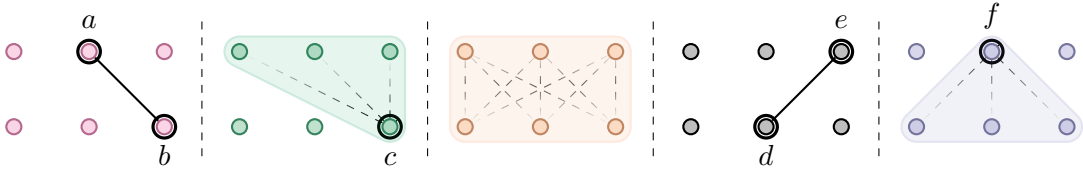
\begin{figure}[t]
    \centering
    \begin{tikzpicture}[scale=1.1]

    \def\colorlist{{
            "VioletRed", "VioletRed", "VioletRed",
            "ForestGreen", "ForestGreen", "ForestGreen",
            "Peach", "Peach", "Peach",
            "black", "black", "black",
            "Periwinkle", "Periwinkle", "Periwinkle",
            "VioletRed", "Sepia", "black", "black", "black"
        }}

    \def\n{14}
    \def\size{0.2cm}

    \foreach \i in {0, 1, ..., \n}{
        \pgfmathparse{\colorlist[\i]}
        \edef\col{\pgfmathresult}
        \node[graph-vert = {\col}{\size}] (u\i) at (\i, 1) {};
        \node[graph-vert = {\col}{\size}] (d\i) at (\i, 0) {};
    }

    \foreach[count = \i] \j / \pos in
        {u1/above, d2/below, d5/below, d10/below, u11/above, u13/above}{
        \draw[very thick] (\j) circle (0.15);
        \node[text height = 0.2cm, \pos = 5pt] at (\j) {$\numletter{\i}$};
    }

    \foreach \i in {2.5, 5.5, 8.5, 11.5}{
        \draw[dashed] (\i, 1.4) -- ++(0, -1.8);        
    }

    \draw[thick] (u1) -- (d2);
    \draw[thick] (d10) -- (u11);

    \foreach \i in {3, 4, 5}{
        \draw[dashed, path fading = north] ([xshift = 5pt, yshift = 5pt]d5) --
            (d5) -- (u\i);
    }

    \foreach \i in {12, 13, 14}{
        \draw[dashed, path fading = south] ([xshift = 5pt, yshift = 5pt]u13) --
            (u13) -- (d\i);
    }

    \foreach \i in {6, 7, 8}{
        \foreach \j in {6, 7, 8}{
            \draw[dashed, path fading = north, path fading = middle]
                ([xshift = 1pt, yshift = 0pt]u\i) -- (u\i) -- (d\j);
        }
    }

    \path (u6) ++(-45:0.2) coordinate (t); 
    \draw[ForestGreen, opacity = 0.2, fill, fill opacity = 0.1, thick]
        (d5) ++({atan(2) - 180}:0.2) arc ({atan(2) - 180}:0:0.2) -- ++(0, 1)
        arc (0:90:0.2) -- ++(-2, 0) arc (90:{atan(2) + 180}:0.2) -- cycle;

    \pic[shift = {(d6)}] {shaded-area = {Peach}{2}{1}};

    \draw[Periwinkle, opacity = 0.2, fill, fill opacity = 0.1, thick]
        (d12) ++(135:0.2) arc (135:270:0.2) -- ++(2, 0) arc (-90:45:0.2) -- ++(-1, 1)
        arc (45:135:0.2) -- cycle;
   
\end{tikzpicture}
    \caption{
        Matching sunflower $\calm' \subseteq \calm$ with core $D = \{ab, de\}$ constructed out of a
        vertex sunflower~$\calv' \subseteq \calv$ with core $K = \bigcap \calv' = K_1 \sqcup K_2$ where
        $K_1 = \{c, f\}$ and $K_2= \{a, b, d, e\}$. The shaded regions indicate where non-core edges of matchings in $\calm'$ can occur.
    }
    \label{fig:m-sunflower}
\end{figure}

\paragraph{Finding a matching sunflower.}
Let $\calm'\subseteq\calm$ be the matchings corresponding to vertex sets~$\calv'$ according to the 1-to-1 correspondence. Let~$D$ be the matching that connects pairs of vertices in~$K_2$ that share a block. The following claim completes the proof.

\pagebreak
\begin{claim}
$\calm'$ is an $\varepsilon$-matching sunflower (with core $D=\bigcap\calm'$). 
\end{claim}
\begin{proof}
First note that $|\calm'|=|\calv'|\geq 2$. Let us check
that $D=\bigcap\calm'$. We have $D\subseteq\bigcap\calm'$ since, for every edge $uv \in D$ and~$M\in\calm'$, the endpoints $u$, $v$ belong to $V(M)$ and share a block, hence $uv \in M$. 
Moreover,
if $uv \in \bigcap \calm'$,
then the endpoints $u, v$
appear in every matching of $\calm'$,
implying
$\{u,v \} \subseteq \bigcap \calv' = K$.
Since $u$ and $v$ share a block,
we get $uv \in D$.
This shows $D = \bigcap \calm'$.

Let $\bm x \sim \{0,1\}^{2n}$ be a uniform colouring. Our goal will be to show
\begin{equation}\label{eq:error} 
    \Pr[\exists M \in \calm', \forall uv \in M \setminus D\colon \bm x_u = \bm x_v] \ge 1 - \eps/2.
\end{equation}
This would conclude the proof, as conditioning on the event ``$\bm x$ has odd many 1s'' (which is what we really care about) can only double the error parameter. To prove \eqref{eq:error}, we show it holds under conditioning on any event ``$\bm x_{K_1} = \alpha$'' where $\alpha\in\{0,1\}^{K_1}$ (which partitions the probability space).

Consider first the simplest case $\bm x'\coloneqq (\bm x\mid \bm x_{K_1}=1_{K_1})$, where we condition all the colours in $K_1$ to be 1. Note that $\bm x$ and $\bm x'$ have the same (uniform) marginal distribution outside $K$. This means we can invoke the robust sunflower property of $\calv'$ for $\bm x'$: with probability $1-\varepsilon/2$ over~$\bm x'=x'$ there exists $V=V(M) \in \calv'$ such that~$t_{V \setminus K}(x') = 1$. We claim that all edges $uv \in M \setminus D$ are coloured 1 under $x'$. Indeed, if $uv \cap K = \emptyset$, then both endpoints are coloured $x'_u=x'_v=1$ by the sunflower property. Otherwise, say, $uv\cap K=\{v\}$. Here one endpoint is coloured $x'_u=1$ by the sunflower property, and the other endpoint has $x'_v = 1$ because of our conditioning.

More generally, we can apply the same logic for $\bm x'\coloneqq (\bm x\mid \bm x_{K_1}=\alpha)$ for any $\alpha\in\{0,1\}^{K_1}$. All we need to do is flip the colours in all blocks that contain $v\in K_1$ with $\alpha_v=0$. Let $x^\alpha \in \{0,1\}^{2n}$ be the fixed colouring that assigns colour 1 to all blocks that contain $v\in K_1$ with $\alpha_v=0$. Formally,~$x^\alpha_v=1$ iff $b(v)\in b(\{u\in K_1:\alpha_u=0\})$, so in particular $x^\alpha_{K_1}\oplus 1_{K_1}=\alpha$. Note that $\bm x'\oplus x^\alpha$ has a uniform marginal distribution outside $K$. This means we can invoke the robust sunflower property of $\calv'$ for $\bm x' \oplus x^\alpha$: with probability $1-\eps/2$ over $\bm x'\oplus x^\alpha=x'\oplus x^\alpha$ there exists $V=V(M)\in\calv'$ such that~$t_{V \setminus K}(x'\oplus x^\alpha) = 1$. We claim that all edges $uv\in M\setminus D$ are monochromatic under $x'$. The case $uv\cap K=\emptyset$ is the same as above. For $uv\cap K=\{v\}$, we have from the sunflower property that~$x'_u\oplus x^\alpha_u=1$. This implies $x'_u=x^\alpha_u\oplus 1= x^\alpha_v\oplus 1=\alpha_v=x'_v$, as desired.
\end{proof}

\section{Approximation method \texorpdfstring{(\cref{thm:main,thm:l,thm:ac0})}{}}
\label{sec:approximation}

Consider the input distribution $\cald \coloneqq (\cald_0+\cald_1)/2$ where $\cald_i$ is supported on $\PM^{-1}(i)$ so that
\begin{itemize}[label=$-$,noitemsep]
    \item $\cald_1$ is the uniform distribution over perfect matchings in $K_{n,n}$;
    \item $\cald_0$ is the odd cut distribution from \Cref{def:odd-cut}.
\end{itemize}
To prove \cref{thm:main}, we start with a small monotone circuit computing $\PM$ and aim for a contradiction. Our goal is to approximate the circuit (relative to $\cald$) with a small monotone DNF
\[
F_\calm \coloneqq \bigvee_{M\in\calm} t_M
\qquad\text{where}\qquad t_M\coloneqq \bigwedge_{e\in M} x_e,
\]
and where $\calm$ is a set of (partial) matchings in $K_{n,n}$. We say that $F=F_\calm$ is \emph{$r$-small} if, for every~$\ell$, $\calm$ contains at most $r^\ell$ matchings of size $\ell$, that is, $|\calm\cap\calp_\ell|\leq r^\ell$ where $\calp_\ell$ is the set of all $\ell$-matchings. The approximation method proceeds in two steps:
\begin{lemma}
    \label{lem:dnfs-approximate-intro}
    Suppose a monotone circuit of size $2^w$ computes $\PM$.
    Then there is an $O(w^3 \log^2 n)$-small
    monotone DNF $F$ such that
    $\Pr_{\bm x \sim \cald}[F(\bm x) = \PM(\bm x)] \geq 1-o(1)$.
\end{lemma}
\begin{lemma}
    \label{lem:dnfs-are-dumb-intro}
    If a monotone DNF $F$ is $o(n)$-small, then
    $\Pr_{\bm x \sim \cald}[F(\bm x) = \PM(\bm x)] \leq
    1/2+o(1)$.
\end{lemma}

These lemmas imply that any monotone circuit of size $S$ for $\PM$ has
$\log^3 S \log^2 n \geq {\Omega}(n)$. Thus $S\geq \exp(n^{1/3-o(1)})$ which proves \cref{thm:main}. It
remains to prove these two lemmas, which we do in~\crefrange{sec:proof-correlation}{sec:proof-approx}. Finally,
we discuss how to derive \crefrange{thm:l}{thm:ac0} in \cref{s:grigni-sipser}.

\subsection{Proof of \texorpdfstring{\Cref{lem:dnfs-are-dumb-intro}}{Lemma~\ref{lem:dnfs-are-dumb-intro}}}
\label{sec:proof-correlation}

The lemma is trivially true if $F=F_\calm$ contains the empty term (so that $F\equiv 1$). Otherwise,
\begin{align*}
    \Pr_{\bm x \sim \cald}[F(\bm x) = \PM(\bm x)]~
    & =~ {\textstyle\frac{1}{2}}\Pr_{\bm x \sim \cald_0}[F(\bm x) = 0]
     + {\textstyle\frac{1}{2}}\Pr_{\bm x \sim \cald_1}[F(\bm x) = 1]\\
    & \leq~ {\textstyle\frac{1}{2}} + 
    \sum_{\ell\in[n]}
    \sum_{M\in\calm\cap\calp_\ell}\Pr_{\bm x\sim\cald_1}[t_M(\bm x)=1] \\
    & \leq~ {\textstyle\frac{1}{2}} + \sum_{\ell\in[n]}
    o(n)^{\ell}(e/n)^\ell \tag{$o(n)$-smallness and \cref{claim:match} below}\\
    & \leq~ {\textstyle\frac{1}{2}} + \sum_{\mathclap{\ell\in\mathbb{N}\setminus\{0\}}}
    o(1)^{\ell} \\
    & \leq~ {\textstyle\frac{1}{2}}+o(1).
\end{align*}

\begin{claim}
    \label{claim:match}
    $\Pr_{\bm x\sim\cald_1}[t_M(\bm x)=1] \leq (e/n)^\ell$ for every $M \in \calp_\ell$.
\end{claim}
\begin{proof}
The probability corresponds to the fraction of perfect matchings of $K_{n,n}$ that contain $M$. This is equal to $(n-\ell)!/n!$ and we verify by induction on $\ell$ and $n$ that $(n-\ell)!/n! \leq (e/n)^\ell$.
    Note that this holds for $\ell = 1$ and all
    $n \in \mathbb{N}$.
    When $n \geq \ell \geq 2$ we obtain
    \begin{equation*}
        \frac{(n-\ell)!}{n!}
        =
        \frac{1}{n}
        \frac{(n-\ell)!}{(n-1)!}
        \leq
        \frac{1}{n}
        \left(
            \frac{e}{n-1}
        \right)^{\ell-1}
        =
        \left( \frac{e}{n} \right)^{\ell}
        \frac{1}{e}
        \left( 1 + \frac{1}{n-1} \right)^{\ell-1}
        \leq
        \left( \frac{e}{n} \right)^{\ell}.
        \qedhere
\end{equation*}
\end{proof}

\subsection{Proof of \texorpdfstring{\Cref{lem:dnfs-approximate-intro}}{Lemma \ref{lem:dnfs-approximate-intro}}}
\label{sec:proof-approx}

\paragraph{Notation.}
Let $C$ be a circuit with $\operatorname{size}(C)\leq 2^w$ computing $\PM$.
We say that 
$F_\calm$ has \emph{width~$k$}
if
every matching in $\calm$ has at most $k$ edges,
and
we say that 
$F_{\calm}$
is a
\emph{$(k,r)$-DNF}
if
$F_{\calm}$ is $r$-small and has width $k$.
We set $\eps \coloneqq n^{-3w}$.
Let $r(\ell, \eps)$ be such that any set $\calm$ of $\ell$-matchings of size
at least $r(\ell,\eps)^\ell$ contains an $\eps$-matching sunflower.
Let~$r \coloneqq \smash{\max_{\ell \in [2w]} r(\ell,\eps)}$.
From \Cref{lem:sunflower-bound}, we
obtain that $r \leq O(w \log^2(w/\eps)) \leq O(w^3 \log^2 n)$. We may assume here that $w^3 \log^2n \leq o(n)$ (so that $r\leq o(n)$) as otherwise the result is trivial.

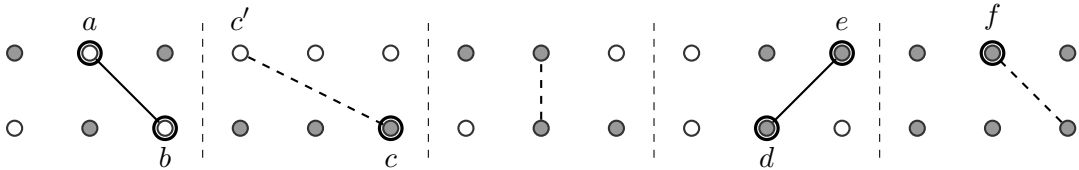
\begin{figure}[b]
    \centering
    \begin{tikzpicture}[scale=1.1]

    \def\colorlistup{{
            "black", "white", "black",
            "white", "white", "white",
            "black", "black", "white",
            "white", "black", "black",
            "black", "black", "black",
        }}
    
    \def\colorlistdown{{
            "white", "black", "white",
            "black", "black", "black",
            "white", "black", "black",
            "white", "black", "white",
            "black", "black", "black",
        }}

    \def\n{14}
    \def\size{0.2cm}

    \foreach \i in {0, 1, ..., \n}{
        \pgfmathparse{\colorlistup[\i]}
        \edef\col{\pgfmathresult}
        \node[graph-vert = {\col}{\size}, black!80, fill = \col!80] (u\i) at (\i, 1) {};
        \pgfmathparse{\colorlistdown[\i]}
        \edef\col{\pgfmathresult}
        \node[graph-vert = {\col}{\size}, black!80, fill = \col!80] (d\i) at (\i, 0) {};
    }

    \foreach[count = \i] \j / \pos in
        {u1/above, d2/below, d5/below, d10/below, u11/above, u13/above}{
        \draw[very thick] (\j) circle (0.15);
        \node[text height = 0.2cm, \pos = 5pt] at (\j) {$\numletter{\i}$};
    }
    \node[text height = 0.2cm, above = 5pt] at (u3) {$c'$};

    \foreach \i in {2.5, 5.5, 8.5, 11.5}{
        \draw[dashed] (\i, 1.4) -- ++(0, -1.8);        
    }

    \draw[thick] (u1) -- (d2);
    \draw[thick] (d10) -- (u11);

    \draw[thick, dashed] (d5) -- (u3);
    \draw[thick, dashed] (d7) -- (u7);
    \draw[thick, dashed] (d14) -- (u13);

\end{tikzpicture}
    \caption{
        Example of 
        an input $x\in\supp(\cald_0)$ (given by the black/white vertex colouring) that causes an error when plucking a sunflower $\calm'$ (from
        \cref{fig:m-sunflower}). Edges $ab$ and $de$ in the core $K=\bigcap\calm'$ are monochromatic, which means $t_K(x)=1$. However, the dashed edges represent a petal~$M\setminus K$, $M\in\calm'$, that contains a bichromatic edge $cc'$, which means $t_M(x)=0$.
    }
    \label{fig:plucking-error}
\end{figure}

\paragraph{Proof overview (via plucking).}
We will construct a $(w,r)$-DNF for $C$
gate-by-gate, inductively, starting at the input gates until we reach the
output gate.
Every input variable is already a~$(w,r)$-DNF.
Consider then an~$\lor$ gate. The challenge is that if we were to naively
combine our inductively constructed $(w,r)$-DNFs by $\lor$, the number of
terms might increase, potentially violating~$r$-smallness. For an~$\land$
gate, a naive combination would also increase the width from $w$ up to~$2w$. In order to maintain smallness of our DNF, we approximate the naive
combination 
by running
\cref{alg:pluck}. The following claim
summarises the properties of the resulting DNF.

\begin{algorithm}[H]
    \caption{Plucking procedure $\pluck(\calm)$}
        \label{alg:pluck}
    \begin{algorithmic}[1]
        \While{$\exists \ell\in [2w]\colon |\calm \cap \calp_\ell| > r^\ell$}
            \State
            Let $\calm' \subseteq \calm \cap \calp_\ell$ be an $\eps$-matching sunflower with core $K$
            \State
            Let
            $\calm \leftarrow (\calm \setminus \calm')\cup \{K\}$
        \EndWhile
    \end{algorithmic}
\end{algorithm}
\begin{claim}
    \label{claim:pluck}
    If $F_{\calm}$ has width $2w$, then
    $F_{\pluck(\calm)}$ is 
    a $(2w,r)$-DNF with $F_{\pluck(\calm)} \geq F_{\calm}$
    and
    \begin{equation}
    \label{eq:pluck}        
    \Pr_{\bm x \sim \cald_0}[F_{\pluck(\calm)}(\bm x) > F_{\calm}(\bm x)]
        \leq
        n^{-w}.
    \end{equation}
\end{claim}
\begin{proof}
First note that Line 2 of the algorithm is always possible as $r \geq
r(\ell,\eps)$ for all $\ell\in[2w]$. Also note that in Line 3, the size of
the family $\calm$ decreases by at least one. This means that the algorithm will
terminate in at most $|\calm|\leq n^{2w}$ iterations. Let us then verify
\eqref{eq:pluck} by calculating the errors incurred in Line 3. An error
occurs for input $x$ only if $t_K(x)=1$ but $t_M(x)=0$ for all~$M\in\calm'$ (see
\cref{fig:plucking-error}). This can occur only if $t_{M\setminus K}(x)=0$
for all $M\in\calm'$. But the $\eps$-matching sunflower
property~\eqref{eq:match-sunflower} for~$\calm'$ implies that this happens
only with probability at most $\eps$ over~$\bm x\sim\cald_0$. A union bound
over all iterations shows that \eqref{eq:pluck} is at most $n^{2w}\eps=n^{2w}n^{-3w}=n^{-w}$.
\end{proof}

\paragraph{Approximation.}
Suppose 
that we have inductively constructed $(w,r)$-DNFs $F_{\calm}$, $F_{\calm'}$ for two gates that feed into a gate $g$ that computes a binary operation $\circ\in\{\lor,\land\}$. Our goal is to find a~$(w,r)$-DNF $F_{\calg}$ that approximates $F_{\calm}\circ F_{\calm'}$ with  tiny error:
\begin{align}
\textstyle \Pr_{\bm x \sim\cald_0}[F_\calg(\bm x) > (F_{\calm}\circ F_{\calm'})(\bm x)] &~\leq~ 2^{-\omega(w)}, \label{eq:0} \\
\textstyle \Pr_{\bm x \sim\cald_1}[F_\calg(\bm x) < (F_{\calm}\circ F_{\calm'})(\bm x)] &~\leq~ 2^{-\omega(w)}. \label{eq:1} 
\end{align}

If $\circ=\lor$, we set $\calg\coloneqq\pluck(\calm \cup \calm')$. To analyse this, we note that plucking incurs $n^{-w}$ errors on $\cald_0$ by \cref{claim:pluck}, which verifies \eqref{eq:0}. On the other hand, plucking introduces no errors on $\cald_1$, verifying~\eqref{eq:1}. If~$\circ=\land$, we first approximate
\begin{equation*}
    \calg' \coloneqq 
    \pluck
    \left(
    \{ M \cup M' : M \in \calm,\, M' \in \calm',\,
    M \cup M' \text{ is a matching}
    \}
\right),
\end{equation*}
and then delete all matchings of size larger than $w$ from $\calg'$; call the resulting family $\calg$.
To analyse this, we note that plucking incurs $n^{-w}$ errors on $\cald_0$
and no errors on $\cald_1$ (we only omitted terms that were not matchings).
Moreover,
deleting wide terms incurs no error on $\cald_0$, 
and the errors on~$\cald_1$ can be bounded as follows:
\begin{equation*}
\Pr_{\bm x\sim\cald_1}[F_{\calg}(\bm x) < F_{\calg'}(\bm x)]
\leq
    \sum_{\ell = w}^{2w} 
    \sum_{M\in\calg'\cap\calp_\ell}\Pr_{\bm x\sim\cald_1}[t_M(\bm x)=1] 
    \stackrel{\text{(\cref{claim:match})}}{\leq}
    \sum_{\ell = w}^{2w} 
    r^\ell \left( e/n \right)^\ell 
    \leq
    \sum_{\ell = w}^{\infty} 
    o(1)^\ell
    \leq
    2^{-\omega(w)}.
\end{equation*}
This verifies \eqref{eq:0}--\eqref{eq:1}. We now conclude the proof by observing that the DNF $F$ for the output gate has tiny overall error, by summing up all the individual contributions in \crefrange{eq:0}{eq:1}:
\begin{align*}
    \Pr_{\bm x \sim \cald}[F(\bm x) \neq \PM(\bm x)]
    =
    {\textstyle\frac{1}{2}}
    \Pr_{\bm x \sim \cald_1}[F(\bm x)<C(\bm x)] 
    + 
    {\textstyle\frac{1}{2}}
    \Pr_{\bm x \sim \cald_0}[F(\bm x)>C(\bm x)] 
    \le 
    \operatorname{size}(C)\cdot 2^{-\omega(w)}
    \le 
    o(1).
\end{align*}

\subsection{Proofs of \texorpdfstring{\Crefrange{thm:l}{thm:ac0}}{Theorems \ref{thm:l} and \ref{thm:ac0}}}
\label{s:grigni-sipser}

To prove \cref{thm:l} we note that the above proof only ever
assumed that the circuit~$C$ computes $\PM$ correctly on the support of
$\cald$. But $\PM(x)=\Odd(x)$ for all $x\in\operatorname{supp}(\cald)$, and
hence the lower bound also applies to $\Odd$.

To prove \Cref{thm:ac0} we apply a standard padding argument and a folklore
depth-reduction result. Indeed, it is known that any function in $\L$ can
be computed by an $\AC^0$-circuit of size $2^{n^\eps}$, where we can take
$\eps>0$ as any fixed constant~\cite[Lemma 8.1]{Allender2008}. Let~$\eps
\coloneqq 1/(4(k+1))$ and define the padded function $f_k\colon \{0,1\}^N \to
\{0,1\}$ by $f_k(x,y) \coloneqq \Odd_n(x)$, where $N \coloneqq 2^{n^\eps}$ and~$\Odd_n$ is odd factor on $n$-vertex graphs. It follows that $f_k$ can be
computed by an~$\AC^0$ circuit of size~$2^{n^\eps} = N$, and its monotone
complexity is at least $\exp(n^{1/3-o(1)}) \geq N^{\Omega(\log^k N)}$.

\section{Discussion}

Let us make some final comments about our proof.
First, quantitatively improved lower bounds for bipartite matching follow immediately from improved bounds on matching sunflowers.
Indeed, our proof in \cref{sec:approximation} shows more generally that monotone circuits of size $O(n/r_w)^w$
can be approximated by $O(r_w)$-small DNFs where $r_w \coloneqq \max_{\ell
\in [2w]} r(\ell,n^{-3w})$. This implies the following closed-form
expression for the lower bound on the monotone complexity of bipartite
matching (also derivable from the ``abstract sunflowers'' of~\cite{Cavalar2020a}):
\begin{equation*}
    \max_{w \in [n]}
    \Omega(n/ r_w)^w.
\end{equation*}
For example, plugging in Razborov's~\cite{Razborov1985} bound  $r(\ell,\eps) \leq (2^\ell \ell \log(1/\eps))^{2\ell}$ would recover his~$n^{\Omega(\log n)}$ lower bound.
For another example, instead of using the optimised bounds for robust
sunflowers~\cite{Alweiss2021,Rao2020,bcw21} in our
\cref{lem:matching-sunflowers}, we could plug in an earlier bound of
Rossman~\cite{Rossman2014}. 
Using Rossman's bound  would already yield an
exponential lower bound for bipartite matching, albeit with a constant smaller than $1/3$ in the exponent.

We also note that our proof extends to other distributions (and functions) than the odd cut distribution $\cald_0$ in \cref{def:odd-cut}. For example, we could generate a bipartite graph out of a random vertex colouring $\bm c$ whose number of 1s on opposite sides of the graph differ by 1. As remarked in~\cite{Feder1998}, these are rejecting inputs for the \emph{$\mathbb{Z}_q$-satisfiability} problem (satisfiability of systems of linear equations modulo $q$). This recovers the exponential monotone circuit lower bound for~$\mathbb{Z}_q$-satisfiability first proved in~\cite{Garg2020,Goos2019}.

Finally, we can ensure in~\cref{thm:ac0} that the $\AC^0$ functions are graph properties (functions that output the same value on isomorphic graphs) by applying a more sophisticated padding argument from~\cite[Lemma
3.6]{Cavalar2023}. This contrasts with results of~\cite{Rossman2008,Rossman2017},
stating that  \emph{homomorphism-preserved} graph properties in $\AC^0$ 
can also be computed by small monotone DNFs.

\bigskip

\subsection*{Acknowledgements}
We thank Susanna de Rezende, Shuo Pang, Jonas Conneryd, and Noel Arteche for
discussions. 
We also thank 
Felix Moreno Peñarrubia and Anup Rao 
for suggestions that simplified and improved the presentation of the paper.
B.C.\ acknowledges support of Royal Society University Research Fellowship URF$\backslash$R1$\backslash$211106
and
EPSRC project EP/Z534158/1 on ``Integrated Approach to Computational Complexity: Structure, Self-Reference and Lower Bounds''. M.G., A.R., A.S., and D.S.\ are supported by the Swiss State
Secretariat for Education, Research, and Innovation (SERI) under contract number
MB22.00026.

\bigskip


\let\OLDthebibliography\thebibliography
\renewcommand\thebibliography[1]{
   \OLDthebibliography{#1}
   \setlength{\parskip}{3pt}
}

\small
\DeclareUrlCommand{\Doi}{\urlstyle{sf}}
\renewcommand{\path}[1]{\footnotesize\Doi{#1}}
\renewcommand{\url}[1]{\href{#1}{\small\Doi{#1}}}
\bibliographystyle{alphaurl}
\bibliography{references.bib}

\newcommand{\etalchar}[1]{$^{#1}$}
\begin{thebibliography}{LMM{\etalchar{+}}22}

\bibitem[AB87]{Alon1987}
Noga Alon and Ravi Boppana.
\newblock The monotone circuit complexity of boolean functions.
\newblock {\em Combinatorica}, 7(1):1--22, 1987.
\newblock \href {https://doi.org/10.1007/BF02579196}
  {\path{doi:10.1007/BF02579196}}.

\bibitem[AB09]{Arora2009}
Sanjeev Arora and Boaz Barak.
\newblock {\em Computational Complexity: A Modern Approach}.
\newblock Cambridge University Press, 2009.
\newblock \href {https://doi.org/10.1017/cbo9780511804090}
  {\path{doi:10.1017/cbo9780511804090}}.

\bibitem[AG87]{Ajtai1987}
Mikl\'{o}s Ajtai and Yuri Gurevich.
\newblock Monotone versus positive.
\newblock {\em Journal of the ACM}, 34(4):1004--1015, 1987.
\newblock \href {https://doi.org/10.1145/31846.31852}
  {\path{doi:10.1145/31846.31852}}.

\bibitem[AHM{\etalchar{+}}08]{Allender2008}
Eric Allender, Lisa Hellerstein, Paul McCabe, Toniann Pitassi, and Michael
  Saks.
\newblock Minimizing disjunctive normal form formulas and {AC}$^0$ circuits
  given a truth table.
\newblock {\em SIAM Journal on Computing}, 38(1):63--84, 2008.
\newblock \href {https://doi.org/10.1137/060664537}
  {\path{doi:10.1137/060664537}}.

\bibitem[ALWZ21]{Alweiss2021}
Ryan Alweiss, Shachar Lovett, Kewen Wu, and Jiapeng Zhang.
\newblock Improved bounds for the sunflower lemma.
\newblock {\em Annals of Mathematics}, 194(3), 2021.
\newblock \href {https://doi.org/10.4007/annals.2021.194.3.5}
  {\path{doi:10.4007/annals.2021.194.3.5}}.

\bibitem[And87]{Andreev1987}
Alexander Andreev.
\newblock A method for obtaining efficient lower bounds for monotone
  complexity.
\newblock {\em Algebra and Logic}, 26(1):1--18, 1987.
\newblock \href {https://doi.org/10.1007/bf01978380}
  {\path{doi:10.1007/bf01978380}}.

\bibitem[ARZ99]{Allender1999}
Eric Allender, Klaus Reinhardt, and Shiyu Zhou.
\newblock Isolation, matching, and counting uniform and nonuniform upper
  bounds.
\newblock {\em Journal of Computer and Systems Sciences}, 59(2):164--181, 1999.
\newblock \href {https://doi.org/10.1006/JCSS.1999.1646}
  {\path{doi:10.1006/JCSS.1999.1646}}.

\bibitem[BCW21]{bcw21}
Tolson Bell, Suchakree Chueluecha, and Lutz Warnke.
\newblock Note on sunflowers.
\newblock {\em Discrete Mathematics}, 344(7):112367, 2021.
\newblock \href {https://doi.org/10.1016/j.disc.2021.112367}
  {\path{doi:10.1016/j.disc.2021.112367}}.

\bibitem[BGW99]{Babai1999}
László Babai, Anna Gál, and Avi Wigderson.
\newblock Superpolynomial lower bounds for monotone span programs.
\newblock {\em Combinatorica}, 19(3):301--319, 1999.
\newblock \href {https://doi.org/10.1007/s004930050058}
  {\path{doi:10.1007/s004930050058}}.

\bibitem[BM25]{Blasiok2025}
Jaros\l{}aw B\l{}asiok and Linus Meierh{\"{o}}fer.
\newblock Hardness of clique approximation for monotone circuits.
\newblock In {\em Proceedings of the 40th Computational Complexity Conference
  (CCC)}, LIPIcs. Schloss Dagstuhl, 2025.
\newblock \href {https://doi.org/10.48550/ARXIV.2501.09545}
  {\path{doi:10.48550/ARXIV.2501.09545}}.

\bibitem[Cav20]{Cavalar2020a}
Bruno Cavalar.
\newblock Sunflower theorems in computational complexity.
\newblock Master's dissertation, Instituto de Matemática e Estatística,
  University of São Paulo, 2020.
\newblock \href {https://doi.org/10.11606/D.45.2020.tde-25112020-162107}
  {\path{doi:10.11606/D.45.2020.tde-25112020-162107}}.

\bibitem[CDM21]{Chat21}
Arkadev Chattopadhyay, Rajit Datta, and Partha Mukhopadhyay.
\newblock Lower bounds for monotone arithmetic circuits via communication
  complexity.
\newblock In {\em Proceedings of 53rd Symposium on Theory of Computing (STOC)},
  pages 786--799. ACM, 2021.
\newblock \href {https://doi.org/10.1145/3406325.3451069}
  {\path{doi:10.1145/3406325.3451069}}.

\bibitem[CHO{\etalchar{+}}22]{Chen2022}
Lijie Chen, Shuichi Hirahara, Igor Oliveira, J{\'a}n Pich, Ninad Rajgopal, and
  Rahul Santhanam.
\newblock Beyond natural proofs: Hardness magnification and locality.
\newblock {\em Journal of the ACM}, 69(4):1--49, 2022.
\newblock \href {https://doi.org/10.1145/3538391} {\path{doi:10.1145/3538391}}.

\bibitem[CKR22]{Cavalar2020}
Bruno Cavalar, Mrinal Kumar, and Benjamin Rossman.
\newblock Monotone circuit lower bounds from robust sunflowers.
\newblock {\em Algorithmica}, 84(12):3655--3685, 2022.
\newblock \href {https://doi.org/10.1007/S00453-022-01000-3}
  {\path{doi:10.1007/S00453-022-01000-3}}.

\bibitem[CO23]{Cavalar2023}
Bruno Cavalar and Igor Oliveira.
\newblock Constant-depth circuits vs. monotone circuits.
\newblock In {\em Proceedings of the 38h Computational Complexity Conference
  (CCC)}, volume 264 of {\em LIPIcs}, pages 29:1--29:37. Schloss Dagstuhl,
  2023.
\newblock \href {https://doi.org/10.4230/LIPIcs.CCC.2023.29}
  {\path{doi:10.4230/LIPIcs.CCC.2023.29}}.

\bibitem[COS17]{Chen2017}
Xi~Chen, Igor Oliveira, and Rocco Servedio.
\newblock Addition is exponentially harder than counting for shallow monotone
  circuits.
\newblock In {\em Proceedings of 49th Symposium on Theory of Computing (STOC)},
  pages 1232--1245. ACM, 2017.
\newblock \href {https://doi.org/10.1145/3055399.3055425}
  {\path{doi:10.1145/3055399.3055425}}.

\bibitem[Dam90]{Damm1990}
Carsten Damm.
\newblock Problems complete for $\oplus${L}.
\newblock {\em Information Processing Letters}, 36(5):247--250, 1990.
\newblock \href {https://doi.org/10.1016/0020-0190(90)90150-v}
  {\path{doi:10.1016/0020-0190(90)90150-v}}.

\bibitem[dRV25]{Rezende2025}
Susanna de~Rezende and Marc Vinyals.
\newblock Lifting with colourful sunflowers.
\newblock In {\em Proceedings of the 40th Computational Complexity Conference
  (CCC)}, LIPIcs. Schloss Dagstuhl, 2025.

\bibitem[FV98]{Feder1998}
Tom{\'{a}}s Feder and Moshe Vardi.
\newblock The computational structure of monotone monadic {SNP} and constraint
  satisfaction: {A} study through datalog and group theory.
\newblock {\em SIAM Journal on Computing}, 28(1):57--104, 1998.
\newblock \href {https://doi.org/10.1137/S0097539794266766}
  {\path{doi:10.1137/S0097539794266766}}.

\bibitem[GGKS20]{Garg2020}
Ankit Garg, Mika G{\"o}{\"o}s, Pritish Kamath, and Dmitry Sokolov.
\newblock Monotone circuit lower bounds from resolution.
\newblock {\em Theory of Computing}, 16(1):1--30, 2020.
\newblock \href {https://doi.org/10.4086/toc.2020.v016a013}
  {\path{doi:10.4086/toc.2020.v016a013}}.

\bibitem[GKRS19]{Goos2019}
Mika G\"{o}\"{o}s, Pritish Kamath, Robert Robere, and Dmitry Sokolov.
\newblock Adventures in monotone complexity and {TFNP}.
\newblock In {\em Proceedings of 10th Innovations in Theoretical Computer
  Science Conference (ITCS)}, volume 124 of {\em LIPIcs}, pages 38:1--38:19.
  Schloss Dagstuhl, 2019.
\newblock \href {https://doi.org/10.4230/LIPIcs.ITCS.2019.38}
  {\path{doi:10.4230/LIPIcs.ITCS.2019.38}}.

\bibitem[GS92]{Grigni1992}
Michelangelo Grigni and Michael Sipser.
\newblock Monotone complexity.
\newblock In {\em Boolean Function Complexity}, London Mathematical Society
  Lecture Note Series, pages 57--75. Cambridge University Press, 1992.
\newblock \href {https://doi.org/10.1017/CBO9780511526633.006}
  {\path{doi:10.1017/CBO9780511526633.006}}.

\bibitem[HR00]{Harnik2000}
Danny Harnik and Ran Raz.
\newblock Higher lower bounds on monotone size.
\newblock In {\em Proceedings of 32nd Symposium on Theory of Computing (STOC)},
  pages 378--387. ACM, 2000.
\newblock \href {https://doi.org/10.1145/335305.335349}
  {\path{doi:10.1145/335305.335349}}.

\bibitem[Juk99]{Jukna1999}
Stasys Jukna.
\newblock Combinatorics of monotone computations.
\newblock {\em Combinatorica}, 19(1):65--85, 1999.
\newblock \href {https://doi.org/10.1007/s004930050046}
  {\path{doi:10.1007/s004930050046}}.

\bibitem[Juk12]{Jukna2012}
Stasys Jukna.
\newblock {\em Boolean function complexity}, volume~27 of {\em Algorithms and
  Combinatorics}.
\newblock Springer, 2012.
\newblock Advances and frontiers.
\newblock \href {https://doi.org/10.1007/978-3-642-24508-4}
  {\path{doi:10.1007/978-3-642-24508-4}}.

\bibitem[LMM{\etalchar{+}}22]{Lovett2022}
Shachar Lovett, Raghu Meka, Ian Mertz, Toniann Pitassi, and Jiapeng Zhang.
\newblock Lifting with sunflowers.
\newblock In {\em Proceedings of 13th Innovations in Theoretical Computer
  Science Conference (ITCS)}. Schloss Dagstuhl, 2022.
\newblock \href {https://doi.org/10.4230/LIPICS.ITCS.2022.104}
  {\path{doi:10.4230/LIPICS.ITCS.2022.104}}.

\bibitem[Lov79]{Lovasz1979}
L{\'{a}}szl{\'{o}} Lov{\'{a}}sz.
\newblock On determinants, matchings, and random algorithms.
\newblock In {\em Proceedings of Fundamentals of Computation Theory (FCT)},
  pages 565--574. Akademie-Verlag, Berlin, 1979.

\bibitem[Mul87]{Mulmuley1987}
Ketan Mulmuley.
\newblock A fast parallel algorithm to compute the rank of a matrix over an
  arbitrary field.
\newblock {\em Combinatorica}, 7(1):101--104, 1987.
\newblock \href {https://doi.org/10.1007/BF02579205}
  {\path{doi:10.1007/BF02579205}}.

\bibitem[Oko82]{Okolnishnikova1982}
Elizaveta Okol'nishnikova.
\newblock The effect of negations on the complexity of realization of monotone
  {B}oolean functions by formulas of bounded depth.
\newblock {\em Metody Diskretnogo Analiza}, (38):74--80, 1982.

\bibitem[PR17]{Pitassi2017}
Toniann Pitassi and Robert Robere.
\newblock Strongly exponential lower bounds for monotone computation.
\newblock In {\em Proceedings of 49th Symposium on Theory of Computing (STOC)},
  pages 1246--1255. ACM, 2017.
\newblock \href {https://doi.org/10.1145/3055399.3055478}
  {\path{doi:10.1145/3055399.3055478}}.

\bibitem[Rao20]{Rao2020}
Anup Rao.
\newblock Coding for sunflowers.
\newblock {\em Discrete Analysis}, 2020(2), 2020.
\newblock \href {https://doi.org/10.19086/da.11887}
  {\path{doi:10.19086/da.11887}}.

\bibitem[Raz85a]{Razborov1985}
Alexander Razborov.
\newblock Lower bounds of monotone complexity of the logical permanent
  function.
\newblock {\em Matematicheskie Zametki}, 37(6):887--900, 942, 1985.
\newblock \href {https://doi.org/10.1007/BF01157687}
  {\path{doi:10.1007/BF01157687}}.

\bibitem[Raz85b]{Razborov1985a}
Alexander Razborov.
\newblock Lower bounds on the monotone complexity of some {B}oolean functions.
\newblock {\em Doklady Akademii Nauk SSSR}, 281(4):798--801, 1985.

\bibitem[Rei08]{Reingold2008}
Omer Reingold.
\newblock Undirected connectivity in log-space.
\newblock {\em Journal of the ACM}, 55(4):1--24, 2008.
\newblock \href {https://doi.org/10.1145/1391289.1391291}
  {\path{doi:10.1145/1391289.1391291}}.

\bibitem[Ros08]{Rossman2008}
Benjamin Rossman.
\newblock Homomorphism preservation theorems.
\newblock {\em Journal of the ACM}, 55(3):1--53, 2008.
\newblock \href {https://doi.org/10.1145/1379759.1379763}
  {\path{doi:10.1145/1379759.1379763}}.

\bibitem[Ros14]{Rossman2014}
Benjamin Rossman.
\newblock The monotone complexity of {$k$}-clique on random graphs.
\newblock {\em SIAM Journal on Computing}, 43(1):256--279, 2014.
\newblock \href {https://doi.org/10.1137/110839059}
  {\path{doi:10.1137/110839059}}.

\bibitem[Ros17]{Rossman2017}
Benjamin Rossman.
\newblock An improved homomorphism preservation theorem from lower bounds in
  circuit complexity.
\newblock In {\em Proceedings of 8th Innovations in Theoretical Computer
  Science Conference (ITCS)}, LIPIcs, pages 27:1--27:17. Schloss Dagstuhl,
  2017.
\newblock \href {https://doi.org/10.4230/LIPIcs.ITCS.2017.27}
  {\path{doi:10.4230/LIPIcs.ITCS.2017.27}}.

\bibitem[RW92]{Raz1992}
Ran Raz and Avi Wigderson.
\newblock Monotone circuits for matching require linear depth.
\newblock {\em Journal of the ACM}, 39(3):736--744, 1992.
\newblock \href {https://doi.org/10.1145/146637.146684}
  {\path{doi:10.1145/146637.146684}}.

\bibitem[Tar88]{Tardos1988}
\'{E}va Tardos.
\newblock The gap between monotone and nonmonotone circuit complexity is
  exponential.
\newblock {\em Combinatorica}, 8(1):141--142, 1988.
\newblock \href {https://doi.org/10.1007/BF02122563}
  {\path{doi:10.1007/BF02122563}}.

\bibitem[Wat25]{Watson2025}
Thomas Watson.
\newblock Complexity in computer science.
\newblock Book draft, 2025.

\bibitem[Weg87]{Wegener1987}
Ingo Wegener.
\newblock {\em The complexity of Boolean functions}.
\newblock Wiley-Teubner, 1987.
\newblock URL:
  \url{https://eccc.weizmann.ac.il/static/books/The_Complexity_of_Boolean_Functions/}.

\bibitem[Wig19]{Wigderson2019}
Avi Wigderson.
\newblock {\em Mathematics and Computation: A Theory Revolutionizing Technology
  and Science}.
\newblock Princeton University Press, 2019.
\newblock \href {https://doi.org/10.2307/j.ctvckq7xb}
  {\path{doi:10.2307/j.ctvckq7xb}}.

\end{thebibliography}
 
\end{document}